\documentclass[runningheads]{llncs}
\usepackage{graphicx}

\usepackage{caption}
\usepackage{subcaption}
\usepackage{multirow}
\usepackage{interval}
\usepackage{subfiles}
\usepackage{graphicx}
\usepackage{booktabs}
\usepackage{amsmath}
\usepackage{amssymb}
\usepackage{makecell}

\usepackage{orcidlink}

\begin{document}
\title{Investigating the Robustness of\\Sequential Recommender Systems\\Against Training Data Perturbations
\thanks{This work was partially supported by projects FAIR (PE0000013) and SERICS (PE00000014) under the MUR National Recovery and Resilience Plan funded by the European Union - NextGenerationEU. Supported also by the ERC Advanced Grant 788893 AMDROMA,  EC H2020RIA project “SoBigData++” (871042), PNRR MUR project  IR0000013-SoBigData.it. This work has been supported by the project NEREO (Neural Reasoning over Open Data) project funded by the Italian Ministry of Education and Research (PRIN) Grant no. 2022AEFHAZ.}
}
\titlerunning{Robustness of SRSs Against Training Perturbations}
%
 \author{Filippo Betello\inst{1}\orcidlink{0009-0006-0945-9688} \and
 Federico Siciliano\inst{1} \orcidlink{0000-0003-1339-6983} \and
 \\ Pushkar Mishra \inst{2} \orcidlink{0000-0002-1653-6198} \and
 Fabrizio Silvestri \inst{1} \orcidlink{0000-0001-7669-9055}}
\authorrunning{F. Betello et al.}
%
\institute{Sapienza University of Rome, Rome, Italy \\ \email{\{betello, siciliano, fsilvestri\}@diag.uniroma1.it} \and
AI at Meta, London, UK
\email{pushkarmishra@meta.com}\\
}
\maketitle

\begin{abstract}
Sequential Recommender Systems (SRSs) are widely employed to model user behavior over time.
However, their robustness in the face of perturbations in training data remains a largely understudied yet critical issue.
A fundamental challenge emerges in previous studies aimed at assessing the robustness of SRSs: the Rank-Biased Overlap (RBO) similarity is not particularly suited for this task as it is designed for infinite rankings of items and thus shows limitations in real-world scenarios. For instance, it fails to achieve a perfect score of 1 for two identical finite-length rankings.
To address this challenge, we introduce a novel contribution: Finite Rank-Biased Overlap (FRBO), an enhanced similarity tailored explicitly for finite rankings. This innovation facilitates a more intuitive evaluation in practical settings.
In pursuit of our goal, we empirically investigate the impact of removing items at different positions within a temporally ordered sequence.
We evaluate two distinct SRS models across multiple datasets, measuring their performance using metrics such as Normalized Discounted Cumulative Gain (NDCG) and Rank List Sensitivity.
Our results demonstrate that removing items at the end of the sequence has a statistically significant impact on performance, with NDCG decreasing up to 60\%. Conversely, removing items from the beginning or middle has no significant effect.
These findings underscore the criticality of the position of perturbed items in the training data. As we spotlight the vulnerabilities inherent in current SRSs, we fervently advocate for intensified research efforts to fortify their robustness against adversarial perturbations.

\end{abstract}
\section{Introduction}
Recommender systems have become ubiquitous in our daily lives \cite{adomavicius2005toward}, playing a key role in helping users navigate the vast amounts of information available online.
Thanks to the global spread of e-commerce services, social media and streaming platforms, recommender systems have become increasingly important for personalized content delivery and user engagement \cite{zhang2019deep}.
In recent years, Sequential Recommender Systems (SRSs) have emerged as a popular approach to modeling user behavior over time \cite{quadrana2018sequence}, leveraging the temporal dependencies in users' interaction sequences to make more accurate predictions.

However, despite their success, the robustness of SRSs against perturbations in the training data remains an open research question \cite{Li_2021}.
In real-world scenarios, disruptions may occur when users employ different services for the same purpose. Data becomes fragmented and divided between a service provider and its competitors in such cases. Nevertheless, the provider must train a recommender system with such incomplete data while ensuring robustness to perturbations.
This challenge is accentuated when we scrutinize previous attempts \cite{oh2022rank} to assess the robustness of SRSs: Rank-Biased Overlap (RBO) \cite{webberrbo}, designed explicitly for infinite lists, reveals limitations when applied to real-world scenarios.

Our experiments revolve around the following research questions:
\begin{itemize}
    \item \textbf{RQ1}: Do changes in the training seed heavily impact rankings?
    \item \textbf{RQ2}: How does the type of removal influence the model's performance?
    \item \textbf{RQ3}: Does more item removed  significantly decrease in performance?
\end{itemize}

\looseness -1 
The contribution of this study is two-fold.
Firstly, we propose the novel Finite Rank-Biased Overlap (FRBO) measure. Unlike RBO, which is only suited for infinite rankings, FRBO is specifically designed to assess the robustness of SRSs within finite ranking scenarios, aligning seamlessly with real-world settings.
Secondly, we empirically assess the impact of item removal from user interaction sequences on SRS performance.
Our investigation shows that the most recent user interaction sequence items are critical for accurate recommendation performance. When these items are removed, there is a significant drop in all metrics.

\section{Related works} \label{sec:related}
\subsection{Sequential Recommender Systems}
SRSs are algorithms that leverage a user's past interactions with items to make personalized recommendations over time and they have been widely used in various applications, including e-commerce \cite{schafer2001commerce,hwangbo2018recommendation}, social media \cite{guy2010social,amato2017recommendation}, and music streaming platforms \cite{schedl2015music,schedl2018current,afchar2022explainability}. Compared to traditional recommender systems, SRSs consider the order and timing of these interactions, allowing for more accurate predictions of a user's preferences and behaviors \cite{wang2019sequential}.

Various techniques have been proposed to implement SRSs.
Initially, Markov Chain models were employed in Sequential Recommendation \cite{fouss2005novel,fouss2005web}, but they struggle to capture complex dependencies in long-term sequences. In recent years, Recurrent Neural Networks (RNNs) have emerged as one of the most promising approaches in this field \cite{donkers2017sequential,hidasi2016sessionbased,quadrana2017personalizing}. These methods encode the users' historical preferences into a vector that gets updated at each time step and is used to predict the next item in the sequence. Despite their success, RNNs may face challenges dealing with long-term dependencies and generating diverse recommendations.
Another approach comes from the use of the attention mechanism \cite{vaswani2017attention}: two different examples are SASRec \cite{kang2018self} and BERT4Rec \cite{sun2019bert4rec} architectures. This method dynamically weighs the importance of different sequence parts to better capture the important features and improve the prediction accuracy.
\looseness -1 Recently, Graph Neural Networks have become popular in the field of recommendations system \cite{wu2022graph,purificato2023sheaf}, especially in the sequential domain \cite{chang2021sequential,fan2021continuous}.

\subsection{Robustness in Sequential Recommender Systems}
\looseness -1 Robustness is an important aspect of SRSs as they are vulnerable to noisy and incomplete data.
Surveys on the robustness of recommender systems \cite{hurley2011robustness,burke2015robust} discussed the challenges in developing robust recommender systems and presented various techniques for improving the robustness of recommender systems.

\looseness -1 Many tools have been developed to test the robustness of different algorithms: \cite{o2004collaborative}, provided both a formalisation of the topic and a framework; the same was done more recently by \cite{ovaisi2022rgrecsys}, who developed a toolkit called RGRecSys which provides a unified framework for evaluating the robustness of SRSs.

Few recent work has focused on studying the problem and trying to increase robustness:
\cite{tang2023improving} focuses on robustness in the sense of training stability while
\cite{oh2022rank} investigated the robustness of SRSs to interaction perturbations. They showed that even small perturbations in user-item interactions can lead to significant changes in the recommendations. They proposed Rank List Sensitivity (RLS), a measure of the stability of rankings produced by recommender systems.

\looseness -1 Our work expands these, making a more accurate investigation of the effect different types of perturbation can have on the models' performance. While \cite{oh2022rank} perturbs a single interaction in the whole dataset, we perturb the sequences of all users and analyze the performance as the number of perturbations changes.
We also provide a theoretical contribution in a new sensitivity evaluation measure for finite rankings, presented in Sec. \ref{sec:metrics}.


\section{Methodology} \label{sec:methods}

\subsection{Setting}
In Sequential Recommendation, each user $u$ is represented by a temporally ordered sequence of items $S_u = (I_1, I_2, ..., I_j, ..., I_{L_u-1}, I_{L_u})$ with which it has interacted, where $L_u$ is the length of the sequence for user $u$. 

User-object interactions in real-world scenarios are often fragmented across services, resulting in a lack of comprehensive data. For example, in the domains of movies and TV shows, a single user may interact with content on TV, in a movie theater, or across multiple streaming platforms. To mimic this real-world scenario in our training data perturbations, we considered three different cases, each removing $n$ items at a specific position in the sequence:

\begin{itemize}
    \item \textbf{Beginning}: $S_u = (I_{n+1}, \dots, I_{L_u-1})$. This represents a user who signs up for a new service, so all his past interactions, i.e., those at the beginning of the complete sequence, were performed on other services.
    \item \textbf{Middle}: $S_u = (I_1, \dots,I_{\left\lfloor\frac{L_u-1-n}{2}\right\rfloor},I_{\left\lfloor\frac{L_u-1+n}{2}\right\rfloor} ..., I_{L_u-1})$. This represents a user who takes a break from using the service for a certain period and resumes using it. Still, any interactions they had during the considered period are not available to the service provider.
    \item \textbf{End}: $S_u = (I_1, \dots, I_{L_u-1-N})$. This represents a user who has stopped using the service, so the service provider loses all the subsequent user interactions. The service provider still has an interest in winning the user back through their platform or other means, such as advertising. Thus, it is essential to have a robust model to continue providing relevant items to the user.
\end{itemize}

with $n \in \{1,2,...,10\}$.
In practice, the data is first separated into training, validation and test set (always composed by $I_{L_u}$). Subsequently, only the training data are perturbed, with a methodology dependent on the scenario considered and the model is then trained on these. The models, trained on data perturbed in a different manner, are therefore always tested on the same data.



\subsection{Metrics}\label{sec:metrics}
\looseness -1 To evaluate the performance of the models, we employ traditional evaluation metrics used for Sequential Recommendation: Precision, Recall, MRR and NDCG.

\looseness -1 Moreover, to investigate the stability of the recommendation models, we employ the Rank List Sensitivity (RLS) \cite{oh2022rank}: it compares two lists of rankings $\mathcal{X}$ and $\mathcal{Y}$, one derived from the model trained under standard conditions 
and the other derived from the model trained with perturbed data.

Therefore, having these two rankings, and a similarity function $sim$ between them, we can formalize the RLS measure as:

\begin{equation}
    \boldsymbol{\mathrm{RLS}} = \frac{1}{|\mathcal{X}|} \sum\limits_{k=1}^{|\mathcal{X}|} \text{sim}(R^{X_k}, R^{Y_k})
\end{equation}

where $X_k$ and $Y_k$ represent the $k$-th ranking inside $\mathcal{X}$ and $\mathcal{Y}$ respectively.

RLS's similarity measure can be chosen from two possible options:
\begin{itemize}
    \item \looseness -1\textbf{Jaccard Similarity (JAC)} \cite{jaccard1912distribution} is a normalized measure of the similarity of the contents of two sets. A model is stable if its Jaccard score is close to 1.
    \begin{equation}
       \boldsymbol{\mathrm{JAC(X,Y)}} = \! \frac{|X \cap Y|}{|X \cup Y|}
    \end{equation}
    \item \textbf{Rank-Biased Overlap (RBO)} \cite{webberrbo} measures the similarity of orderings between two rank lists. Higher values indicate that the items in the two lists are arranged similarly:
    \begin{equation} \label{eq:RBOinfinite}
        \boldsymbol{\mathrm{RBO(X,Y)}} = (1-p) \sum_{d=1}^{+\infty} p^{d-1} \frac{|X[1:d] \cap Y[1:d]|}{d}
    \end{equation}
\end{itemize}

\looseness -1 In the domain of recommendation systems, it is customary to compute metrics using finite-length rankings, typically denoted by appending ``@k'' to the metric's name, such as NDCG@k. While traditional metrics (e.g. NDCG, MRR, etc.) readily adapt to finite-length rankings, maintaining their core meaning, the same behaviour does not extend to RLS when employing RBO. The reason lies in Equation \ref{eq:RBOinfinite}, which exhibits a notable limitation: it fails to converge to one, even when applied to identical finite-length lists. To overcome this limitation, we introduce the Finite Rank-Biased Overlap (FRBO) similarity, denoted as $\mathrm{FRBO@k}$, which represents a novel formulation engineered to ensure convergence to a value of 1 for identical lists and a value of 0 for entirely dissimilar lists.
\begin{theorem}\label{def:FRBOfull}
\looseness -1 Given a set of items $I = \{I_1,...,I_{N_I}\}$, two rankings $X = (x_1,...,x_k)$ and $Y= (y_1,...,y_k)$, such that $x_i,y_i \in I$, and $k \in \mathbb{N^+}$
    \begin{equation}\label{eq:FRBOfull}
    \boldsymbol{\mathrm{FRBO(X,Y)@k}} = \frac{\mathrm{RBO(X,Y)@k}-min_{X,Y}\mathrm{RBO@k}}{max_{X,Y}\mathrm{RBO@k}-min_{X,Y}\mathrm{RBO@k}}
    \end{equation}
$$\min_{X,Y}\mathrm{FRBO(X,Y)@k} = 0,\; \max_{X,Y}\mathrm{FRBO(X,Y)@k} = 1$$
\end{theorem}
\begin{proof}
    This follows from the fact that given a function $f:A\rightarrow[a,b]$ and another function $g=\frac{f-a}{b-a}$, then $g:A\rightarrow[0,1]$, where $A$ is any set.
\qed
\end{proof}

\looseness -1 To normalize RBO, we need to identify its minimum and maximum values when the summation is carried out up to the top-k elements of the ranking, simultaneously proving that these values are not naturally constrained to be 0 and 1.

\begin{lemma}\label{thr:RBOmin}
Given a set of items $I = \{I_1,...,I_{N_I}\}$, two rankings $X = (x_1,...,x_k)$ and $Y= (y_1,...,y_k)$, such that $x_i,y_i \in I$, and $k \in \mathbb{N^+}$, the following holds:
\begin{align} \label{eq:RBOmin}
    &\min_{X,Y}\mathrm{RBO@k} =
    \begin{cases}
        0, & \text{if } k \leq \lfloor{\frac{N_I}{2}}\rfloor\\
        (1-p)\left(2\frac{p^{\lfloor\frac{N_I}{2}\rfloor}-p^{N_I}}{1-p}  -N_I\ell\right) & \text{otherwise}
    \end{cases}\\
    &\text{where } \mathrm{RBO(X,Y)@k} = (1-p) \sum_{d=1}^{k} p^{d-1} \frac{|X[1:d] \cap Y[1:d]|}{d} \nonumber\\
    &\text{and }\ell = p^{\left\lfloor\frac{N_I}{2}\right\rfloor}\Phi(p,1,\lfloor\frac{N_I}{2}\rfloor+1) - p^{N_I}\Phi(p,1,N_I+1) \nonumber
    \end{align}
\end{lemma}

\begin{proof}
For the first part, it suffices to consider two rankings $X$, $Y$ that share no common elements, so that $|X[1:d] \cap Y[1:d]|=0$. This leads to:
$$\min_{X,Y}\mathrm{RBO@k} = (1-p)\sum_{d=1}^{k} p^{d-1}\frac{0}{d} = 0$$

However, since the number of items $I$ is not infinite, if $k>\lfloor\frac{N_I}{2}\rfloor$, at least one element must necessarily be in common between the two rankings: $$|X[1:k] \cap Y[1:k]| = 0 \iff |X[1:k] \cup Y[1:k]| = 2k \leq N_I \iff k\leq \lfloor\frac{N_I}{2}\rfloor$$

Consequently, the similarity can't assume a value of 0 if $k>\lfloor\frac{N_I}{2}\rfloor$.

\looseness -1 Given that $\mathrm{RBO(X,Y)@k+1} \geq \mathrm{RBO(X,Y)@k}\,\forall k \in \mathbb{N^+}$, similarity it's minimized when intersections between the two rankings occur as far down the list as possible.
When $d>\lfloor\frac{N_I}{2}\rfloor$, there are at least $2d-N_I$ intersections, because items in $d$-th position in one ranking are necessarily contained in the other.
\begin{align*}
    \min_{X,Y}\mathrm{RBO@k} &= (1-p)\left(0 + \sum_{d=\lfloor\frac{N_I}{2}\rfloor+1}^{N_I}p^{d-1}\frac{2d-N_I}{d}\right)\\
    &= (1-p)\left(2\sum_{d=\lfloor\frac{N_I}{2}\rfloor+1}^{N_I}p^{d-1}  -N_I\sum_{d=\lfloor\frac{N_I}{2}\rfloor+1}^{N_I}\frac{p^{d-1}}{d}\right)
\end{align*}

The first series can be regarded as a finite geometric series, for which we can apply the formula for the sum of a geometric series: $\sum_{n=1}^{k}ar^{n-1} = \frac{a(1-r^k)}{1-r}$:
\begin{align*}
\sum_{d=\lfloor\frac{N_I}{2}\rfloor+1}^{N_I}p^{d-1} = p^{\lfloor\frac{N_I}{2}\rfloor}\sum_{d=1}^{N_I-\lfloor\frac{N_I}{2}\rfloor}p^{d-1} = \frac{p^{\lfloor\frac{N_I}{2}\rfloor}-p^{N_I}}{1-p}
\end{align*}

 \looseness -1 Using Lerch transcendent function $\Phi(z,s,\alpha) \!\!=\!\! \sum_{n=0}^{+\infty} \! \!\frac{z^n}{(n+\alpha)^s}$ in second series:
\begin{align*}
\sum_{d=\lfloor\frac{N_I}{2}\rfloor+1}^{N_I}\frac{p^{d-1}}{d} &= \sum_{d=0}^{N_I-\lfloor\frac{N_I}{2}\rfloor-1}\frac{p^{d+\lfloor\frac{N_I}{2}\rfloor}}{d+\lfloor\frac{N_I}{2}\rfloor+1}\\
&= \sum_{d=0}^{+\infty}\frac{p^{d+\lfloor\frac{N_I}{2}\rfloor}}{d+\lfloor\frac{N_I}{2}\rfloor+1} - \sum_{d=N_I-\lfloor\frac{N_I}{2}\rfloor}^{+\infty}\frac{p^{d+\lfloor\frac{N_I}{2}\rfloor}}{d+\lfloor\frac{N_I}{2}\rfloor+1}\\
&= p^{\lfloor\frac{N_I}{2}\rfloor}\Phi(p,1,\lfloor\frac{N_I}{2}\rfloor+1) - p^{N_I}\Phi(p,1,N_I+1)
\end{align*}
\qed
\end{proof}

\begin{lemma}\label{thr:RBOmax}
Given a set of items $I = \{I_1,...,I_{N_I}\}$, two rankings $X = (x_1,...,x_k)$ and $Y= (y_1,...,y_k)$, such that $x_i,y_i \in I$, and $k \in \mathbb{N^+}$, the following holds:
    \begin{equation} \label{eq:RBOmax}
        \max_{X,Y}\mathrm{RBO@k} = 1-p^k\\
    \end{equation}
\end{lemma}

\begin{proof}
RBO@k reaches its maximum value when the two rankings are identical:
$$|X[1:d] \cap Y[1:d]| = d \,\,\, \forall d \in \{1,\dots,k\}$$

Referring to the result for the geometric series, we can compute:
\begin{align*}
max_{X,Y}\mathrm{RBO@k} &= (1-p)\sum_{d=1}^{k}p^{d-1}\frac{d}{d} = (1-p)\sum_{d=1}^{k}p^{d-1} = 1-p^k\\
\end{align*}
\qed
\end{proof}

\begin{corollary}\label{thr:RBOminmax}
Given a set of items $I = \{I_1,...,I_{N_I}\}$ and $k \in \mathbb{N^+}$:
    $$\min_{X,Y}\mathrm{RBO(X,Y)@k} \neq 0,\; \max_{X,Y}\mathrm{RBO(X,Y)@k} \neq 1$$
    $$\exists X,Y \in \{(x_1,\dots,x_k)|x_i \in I \land x_i \neq x_j \forall i \neq j\}$$
    $$\text{s.t. } \mathrm{RBO(X,Y)@k} \neq \mathrm{FRBO(X,Y)@k}$$
\end{corollary}

\begin{proof}
This follows from Theorems \ref{thr:RBOmin} and \ref{thr:RBOmax}.
\qed
\end{proof}

\looseness -1 It's worth considering that in real-world scenarios, the number of possible items $N_I$ is significantly larger compared to the length $k$ of the rankings used to compute the metrics, i.e., $N_I>>k$. In this context, we can safely omit the minimum value from Eq. \ref{eq:FRBOfull}, resulting in:

\begin{equation*}\label{eq:FRBOsimple}
    \boldsymbol{\mathrm{FRBO(X,Y)@k}} = \frac{\mathrm{RBO(X,Y)@k}}{\max_{X,Y}\mathrm{RBO@k}} = \frac{1-p}{1-p^k}\sum_{d=1}^{k} p^{d-1} \frac{|X[1:d] \cap Y[1:d]|}{d}
\end{equation*}

\looseness -1 In this section, we have shown that RBO is not an adequate similarity score when dealing with finite-length rankings. So, we have derived expressions that quantify the minimum and maximum values of RBO, allowing us to compute a normalized version of RBO.

Kendall's Tau \cite{kendall1948rank} assumes that two rankings contain precisely the same items. However, this assumption may not hold for finite top-k ranked lists. In addition, average overlap \cite{wu2003methods,fagin2003comparing} has a peculiar property of monotonicity in depth, where greater agreement with a deeper ranking does not necessarily lead to a higher score, and less agreement does not necessarily lead to a lower score \cite{webberrbo}.

\section{Experiments} \label{experiments}

\subsection{Datasets}

\begin{table}[t]
    \caption{Dataset statistics after preprocessing}
    \begin{tabular}{l||ccc|cc}
        \toprule
        Dataset & Users & Items & Interactions & Average $\frac{\mathrm{Actions}}{\mathrm{User}}$ & Median $\frac{\mathrm{Actions}}{\mathrm{User}}$ \\
        \midrule
        \midrule
        MovieLens 1M  & 6040  & 3952  & 1M  & 165  & 96   \\
        \midrule
        MovieLens 100K & 943 & 1682 & 100K & 106 & 65  \\
        \midrule
        Foursquare Tokyo & 2293 & 61858 & 537703 & 250 & 153       \\
        \midrule
        Foursquare New York & 1083 & 38333 & 227428 & 210 & 173       \\
        \bottomrule         
    \end{tabular}
    \label{tab:stats}
\end{table}


We use four different datasets:

\textbf{MovieLens} \cite{harpermovielens} $\rightarrow$ This benchmark dataset is often used to test recommender systems. In this work, we use the 100K version
and 1M version.

\textbf{Foursquare} \cite{yang2014modeling} $\rightarrow$ This dataset contains check-ins from New York City and Tokyo collected over a period of approximately ten months.

The statistics for all the datasets are shown in Table \ref{tab:stats}.

We select datasets widely used in the literature and with a high number of interactions per user.
The limitation in dataset selection arises from our intention to assess the robustness against the removal of up to 10 elements.
Therefore, the dataset must satisfy the following constraint: $L_u > 10 \quad \forall u \in U$,
where $L_u$ is the number of interactions of user $u$, i.e. the length of the sequence $S_u$ of interactions, and $U$ is the set of all users in the dataset.
If the condition is not met, we delete all the items for a user with less than ten interactions. In this case, we cannot train the model on this particular user. We have, thus, decided to exclude datasets such as Amazon \cite{ni2019justifying} for they cannot meet the previous criteria.

\subsection{Architectures}
In our study, we use two different architectures to validate the results:

\begin{itemize}
    \item \textbf{SASRec} \cite{kang2018self} uses self-attention processes to determine the importance of each interaction between the user and the item.
    
    \item \textbf{GRU4Rec} \cite{hidasi2016sessionbased} is a recurrent neural network architecture that uses gated recurrent units (GRUs) \cite{cho2014learning} to improve the accuracy of the prediction.
\end{itemize}

\looseness -1 We choose to use these two models because both have demonstrated excellent performance in several benchmarks and have been widely cited in the literature.
Furthermore, as one model employs attention while the other utilizes RNN, their network functioning differs, which makes evaluating their behavior under training perturbations particularly interesting.
We use the models' implementation provided by the RecBole Python library \cite{zhao2021recbole}, with their default hyperparameters.

\subsection{Experimental Setup}
\looseness -1 All the experiments are performed on a single NVIDIA RTX 3090.
The batch size is fixed to 4096. Adam optimizer is used with a fixed learning rate of $5*10^{-4}$. The number of epochs is set to 300, but in order to avoid overfitting, we stop the training when the NDCG@20 does not improve for 50 epochs. The average duration of each run is 1.5 hours.
The RecBole library was utilized for conducting all the experiments, encompassing data preprocessing, model configuration, training, and testing. This comprehensive library ensures the reproducibility of the entire process. All the evaluation metrics are calculated with a cut-off $K$ of 20.
To validate the effective degradation in performance and rankings similarity, we employed the paired Student's t-test \cite{student1908probable}, after testing normality of distributions with Shapiro–Wilk test \cite{shapiro1965analysis}, and a significance level of $10^{-3}$.

\section{Results} \label{sec:results}

\subsection{Intrinsic Models Instability (RQ1)}
\begin{table}[t]
\caption{\textbf{Variation of metrics between two seeds.}
Metrics for the 4 datasets considered for GRU4Rec and SASRec.
For Precision (Prec.), Recall, MRR and NDCG, it is shown the percentage variation between the obtained performance using two different initialization seeds for the models.
For each metric, the value corresponding to the dataset where a model is less robust is highlighted in \textbf{bold}. For each dataset, the value corresponding to the model that is the least robust of the two given a metric is \underline{underlined}.}
\centering
\resizebox{\textwidth}{!}{
\begin{tabular}{l||cccc|cc||cccc|cc}
\toprule
& \multicolumn{6}{c||}{SASRec} & \multicolumn{6}{c}{GRU4Rec} \\
\midrule
 & Prec. & Recall & MRR & NDCG & FRBO & JAC & Prec. & Recall & MRR & NDCG & FRBO & JAC\\
\midrule\midrule
ML 100k & 0.5\% & 0.5\% & 0.5\% & \underline{0.3\%} & .466 & .489 & \underline{\textbf{1.7\%}} & \underline{\textbf{1.7\%}} & \underline{1.2\%} & 0.1\% & \underline{.337} & \underline{.413}\\ 
\midrule
ML 1M & \underline{0.3\%} & \underline{0.4\%} & \textbf{0.5\%} & \textbf{0.5\%} & .549 & .569 & 0.2\% & 0.2\% & \underline{3.5\%} & \underline{2.5\%} & \underline{.311} & \underline{.347}\\
\midrule
FS NYC & \underline{0.9\%} & \underline{0.9\%} & 0.0\% & \underline{0.3\%} & \textbf{.398} & \textbf{.273} & 0.1\% & 0.1\% & \underline{0.6\%} & 0.1\% & \underline{\textbf{.110}} & \underline{\textbf{.083}}\\
\midrule
FS TKY & \underline{\textbf{1.4\%}} & \underline{\textbf{1.4\%}} & 0.4\% & 0.02\% & .418 & .267 & 0.8\% & 0.8\% & \underline{\textbf{4.7\%}} & \underline{\textbf{3.4\%}} & \underline{.210} & \underline{.165}\\
\bottomrule
\end{tabular}
}
\label{tab:baseline_metrics}
\end{table}

\looseness -1 To measure the inherent robustness of the models, i.e. in the baseline case (without removal of items), we train the model twice using different initialization seeds and compute the percentage discrepancy between the Precision, Recall, MRR and NDCG obtained by the two rankings. The results are shown in Table \ref{tab:baseline_metrics}: it can be seen that in general the discrepancy is negligible, almost always less than 1\%.
On the other hand, the two RLS, calculated using FRBO and Jaccard respectively, show us the similarity between the two rankings produced with different initialization seeds. The results deviate significantly from the ideal value of 1, indicating considerably different rankings.
These combined results indicate to us that the models converge to an adequate performance beyond the initialization seed, but that the actual rankings produced are heavily influenced by it.
\looseness -1 The \textbf{bold} represent the dataset with the least robust result for each metric. No datasets stands out significantly, yet Foursquare Tokyo seems to give more problems regarding standard evaluation metrics, while for Foursquare New York City it seems more difficult to produce stable rankings.
The \underline{underlined} values instead compare, for each metric and dataset, which of the two models is the least robust. If we consider metrics that do not consider the position of the positive item in the ranking, i.e. Precision and Recall, GRU4Rec seems more robust than SASRec. If, on the other hand, we look at metrics that penalize relevant items in positions too low in the ranking, we see that the opposite happens. This suggests to us that GRU4Rec can return a better set of results, but in a less relevant order than SASRec does.
As proof of this, if we check the RLS values, we see that GRU4Rec is always the least robust model as the initialization seed changes.


\subsection{Comparison of the position of removal (RQ2)}
Table \ref{tab:comparison_removal_position} compares performance and stability when ten items are removed from the sequence versus retaining all items (reference value) in the training set with a consistent initialization seed. It's observed that discarding items from the beginning or middle of the sequence does not significantly impair the model's performance; only a minor decline is noted, potentially attributable to the marginally reduced volume of total training data.
Instead, it can be observed how removing items from the end of the sequence leads to a drastic reduction in metrics: in the case of SASRec applied to the MovieLens 1M dataset, the NDCG more than halves.
Finally, we can see how the difference between the three settings, although maintaining the same trend, is less marked for GRU4Rec applied to the Foursquare NYC dataset. This may be due to the generally higher performance of GRU4Rec \cite{kang2018self}. A more in-depth analysis is presented in Section \ref{res:dataset_diff}.

\looseness -1 Table \ref{tab:comparison_removal_position} also shows the RLS values, computed using FRBO and Jaccard similarity, on the same model-dataset pairs: removing items at the end of the sequence leads to considerable variation in the rankings produced by the models. The values approach 0, meaning that the produced rankings share almost no items.
Our results are in contrast to those of \cite{oh2022rank}, which instead claim that an initial perturbation of a user sequence leads to a higher impact on the RLS. However, their experimental setting is different than ours, as explained in Section \ref{sec:related}.

\begin{table}[t]
\caption{\looseness -1 \textbf{Variations in metrics for ten-item removal.} Metrics for the 3 scenarios considered for SASRec on ML-1M and GRU4Rec on FS-NYC. For Precision, Recall, MRR, and NDCG, it is shown the percentage variation between removing ten items and the reference value. 
For each metric, in \textbf{bold} it is highlighted the value representing the less robust model. $^\dagger$ indicates a statistically significant result.}
\label{tab:comparison_removal_position}
\centering
\begin{tabular}{c|l||cccc|cc}
\toprule
Model & Removal & Prec. & Recall & MRR & NDCG & FRBO & JAC \\
\midrule
\multirow{3}{*}{\makecell{SASRec\\ML-1M}} & Beginning & -0.23\% & -0.23\% & -0.15\% & -0.07\% & .399$^\dagger$ & .368$^\dagger$ \\
\cline{2-8}
& Middle & -0.35\% & -0.29\% & -1.09\% & -0.73\% & .385$^\dagger$ & .356$^\dagger$ \\
\cline{2-8}
& End & \textbf{-15.5\%$^\dagger$} & \textbf{-15.3\%$^\dagger$} & \textbf{-45.9\%$^\dagger$} & \textbf{-56.0\%$^\dagger$} &\textbf{.080$^\dagger$} & \textbf{.106$^\dagger$} \\
\midrule
\multirow{3}{*}{\makecell{GRU4Rec\\FS-NYC}} & Beginning & -0.23\% & -0.23\% & -1.47\% & -0.94\% & .105$^\dagger$ & .075$^\dagger$ \\
\cline{2-8}
& Middle & -0.93\% & -0.93\% & -0.18\% & -0.44\% & .110$^\dagger$ & .074$^\dagger$ \\
\cline{2-8}
& End & \textbf{-4.92\%$^\dagger$} & \textbf{-4.86\%$^\dagger$} & \textbf{-8.42\%$^\dagger$} & \textbf{-7.39\%$^\dagger$} & \textbf{.089$^\dagger$} & \textbf{.062$^\dagger$} \\
\bottomrule
\end{tabular}
\end{table}

\subsection{Effect of the number of elements removed (RQ3)}
\begin{figure}[!ht]
     \centering
     \begin{subfigure}[t]{0.47\textwidth}
         \centering
         \includegraphics[width=0.86\textwidth]{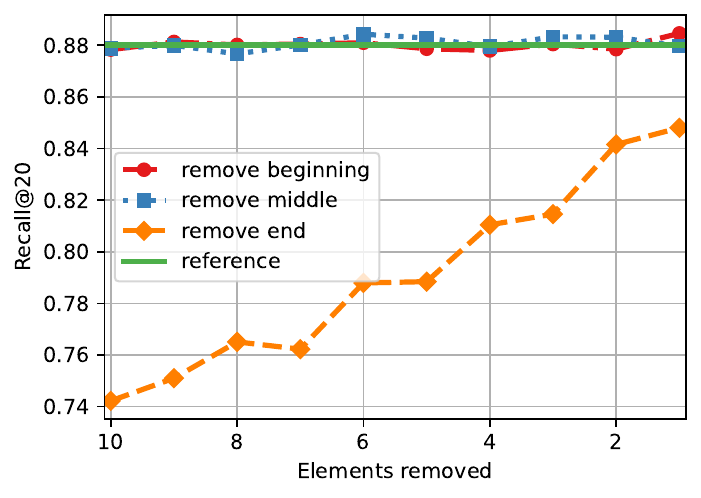}
         \caption{HR@20 SASRec ML-100k}
         \label{fig:number_removal_comparison_recall_sas_ml-1m}
     \end{subfigure}
     \hfill
     \begin{subfigure}[t]{0.47\textwidth}
         \centering
         \includegraphics[width=0.86\textwidth]{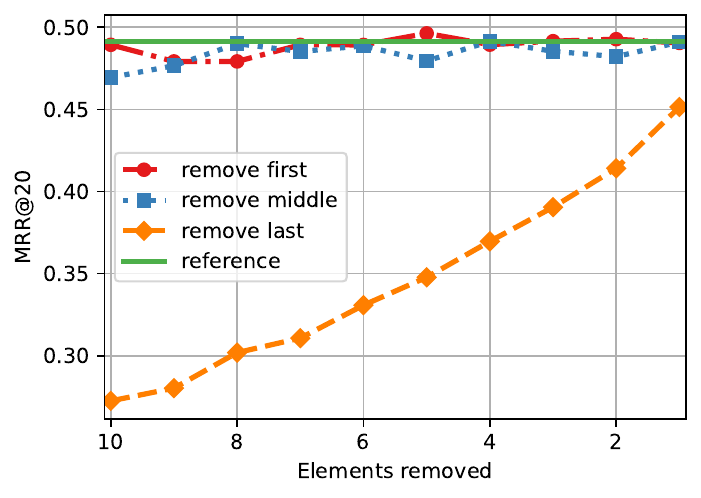}
         \caption{MRR GRU4Rec ML-1M}
         \label{fig:number_removal_comparison_mrr_gru_ml-100k}
     \end{subfigure}
     \\
     \begin{subfigure}[t]{0.47\textwidth}
         \centering
         \includegraphics[width=0.86\textwidth]{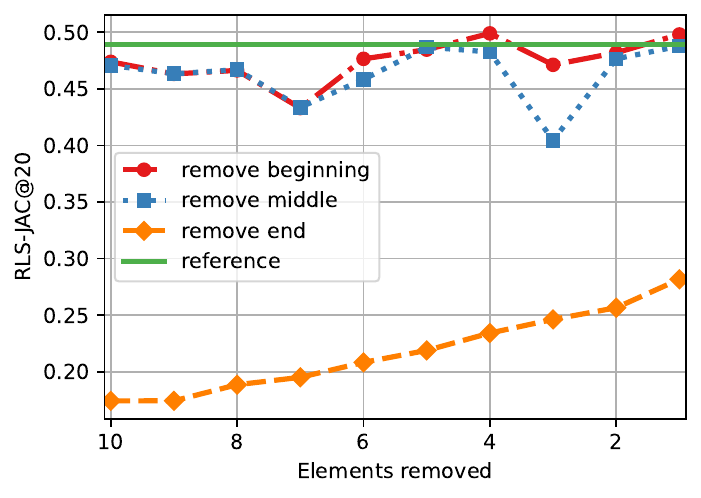}
         \caption{RLS-JAC@20 SASRec ML-100k}
         \label{fig:number_removal_comparison_jac_sas_ml-100k}
     \end{subfigure}
     \hfill
     \begin{subfigure}[t]{0.47\textwidth}
         \centering
         \includegraphics[width=0.86\textwidth]{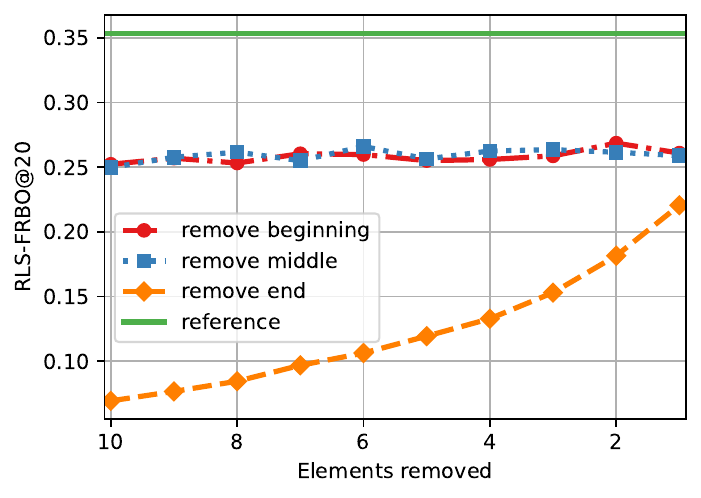}
         \caption{RLS-FRBO@20 GRU4Rec ML-1M}
         \label{fig:number_removal_comparison_rbo_gru_ml-1m}
     \end{subfigure}
    \caption{\looseness -1 Plots of various metrics for the ML-100k and ML-1M datasets as the number of removed elements increases. The baseline is shown as a horizontal solid line, while dashed lines show the metrics as the number of items removed changes for the three scenarios considered.}
    \label{fig:number_removal_comparison}
\end{figure}

\looseness -1 As we discussed in the previous section, removing elements that are at the beginning or in the middle of the temporally ordered sequence has no effect on performance. This is also confirmed by Figure \ref{fig:number_removal_comparison}, where we can also see, however, that for the above-mentioned cases there is no variation as the number of removed elements increases.
The deviation of the RLS displayed in Figure \ref{fig:number_removal_comparison_rbo_gru_ml-1m} will be analyzed in more detail in Section \ref{res:dataset_diff}.
For the remaining setting, the one where we remove items at the end of the sequence, the effect of the number of items removed is evident: the metrics drop drastically as the number of items removed increases. This result holds true for both models considered and for all four datasets tested.

\subsection{Differences between the datasets}\label{res:dataset_diff}

\begin{figure}[!ht]
     \centering
     \begin{subfigure}[t]{0.47\textwidth}
         \centering
         \includegraphics[width=0.86\textwidth]{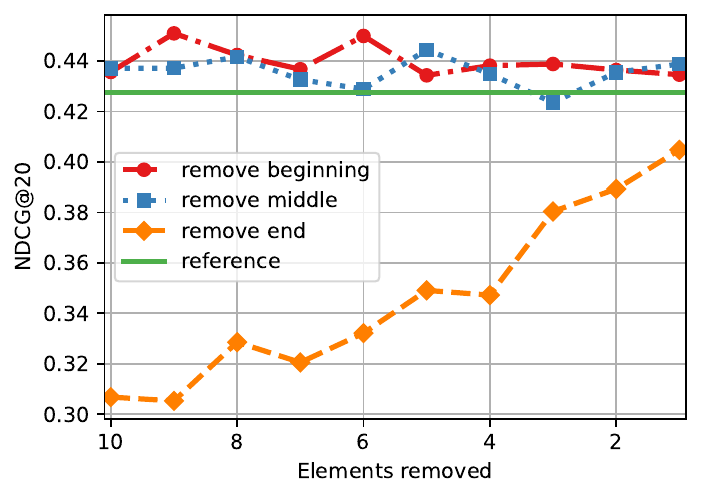}
          \caption{NDCG@20 SASRec ML-100k}\label{fig:datasets_comparison_ndcg_ml-100k}
     \end{subfigure}
     \hfill
     \begin{subfigure}[t]{0.47\textwidth}
         \centering
         \includegraphics[width=0.86\textwidth]{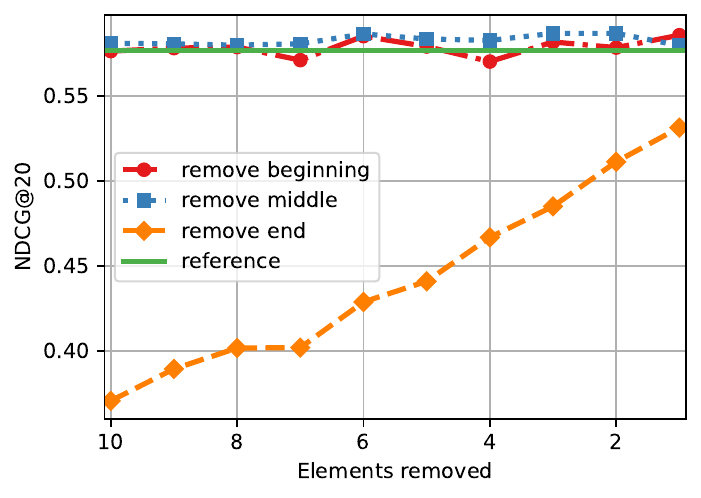}
          \caption{NDCG@20 SASRec ML-1M}
          \label{fig:datasets_comparison_ndcg_ml-1m}
     \end{subfigure}
     \\
     \begin{subfigure}[b]{0.47\textwidth}
         \centering
         \includegraphics[width=0.86\textwidth]{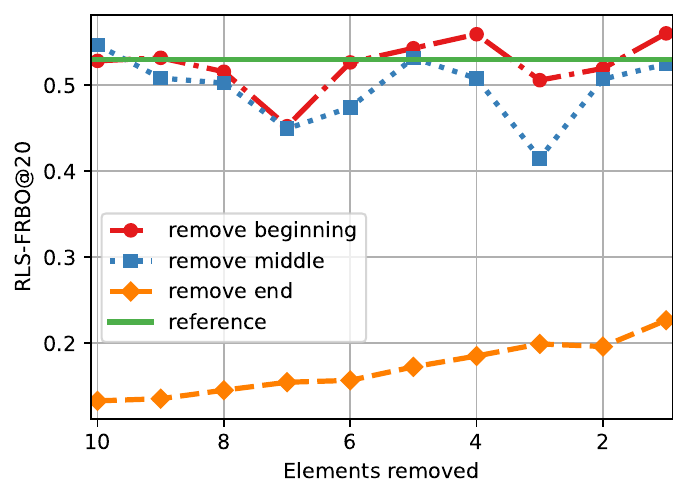}
          \caption{FRBO@20 SASRec ML-100k}
          \label{fig:datasets_comparison_rbo_ml-100k}
     \end{subfigure}
     \hfill
     \begin{subfigure}[b]{0.47\textwidth}
         \centering
         \includegraphics[width=0.86\textwidth]{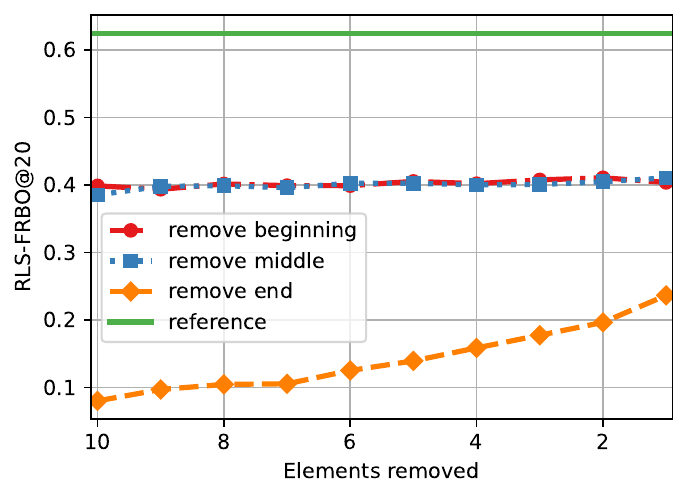}
          \caption{FRBO@20 SASRec ML-1M}
          \label{fig:datasets_comparison_rbo_ml-1m}
     \end{subfigure}
    \caption{\looseness -1 Plots of NDCG and FRBO for SASRec on the ML-100K and ML-1M datasets. The baseline is shown as a horizontal solid line, while dashed lines show the metrics as the number of items removed changes for the three scenarios considered.}    
    \label{fig:datasets_comparison_movielens}
\end{figure}

\begin{figure}[!ht]
     \centering
     \begin{subfigure}[t]{0.47\textwidth}
         \centering
         \includegraphics[width=0.9\textwidth]{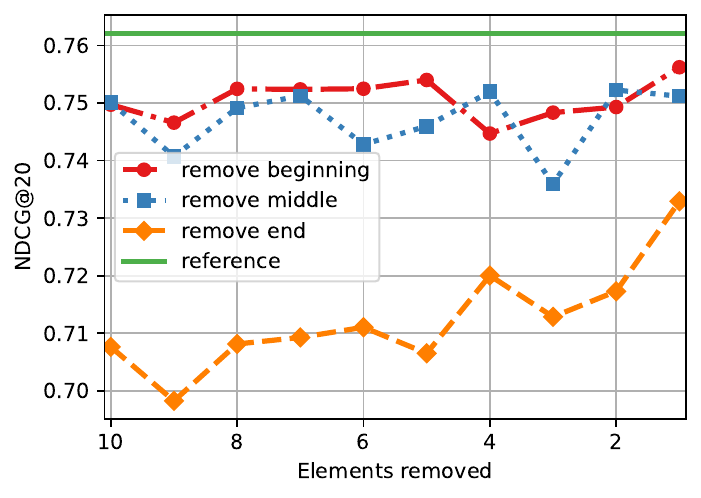}
          \caption{NDCG@20 SASRec FS-TKY}
          \label{fig:datasets_comparison_ndcg_tky}
     \end{subfigure}
     \hfill
     \begin{subfigure}[t]{0.47\textwidth}
         \centering
         \includegraphics[width=0.9\textwidth]{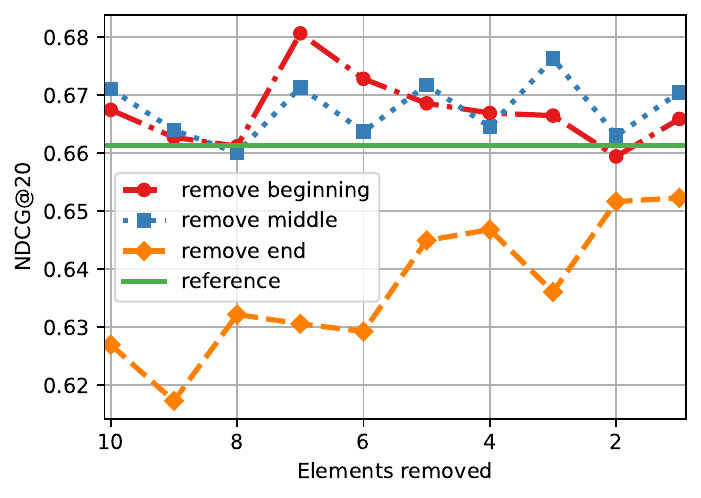}
          \caption{NDCG@20 SASRec FS-NYC}
          \label{fig:datasets_comparison_ndcg_nyc}
     \end{subfigure}
     \\
     \begin{subfigure}[b]{0.47\textwidth}
         \centering
         \includegraphics[width=0.9\textwidth]{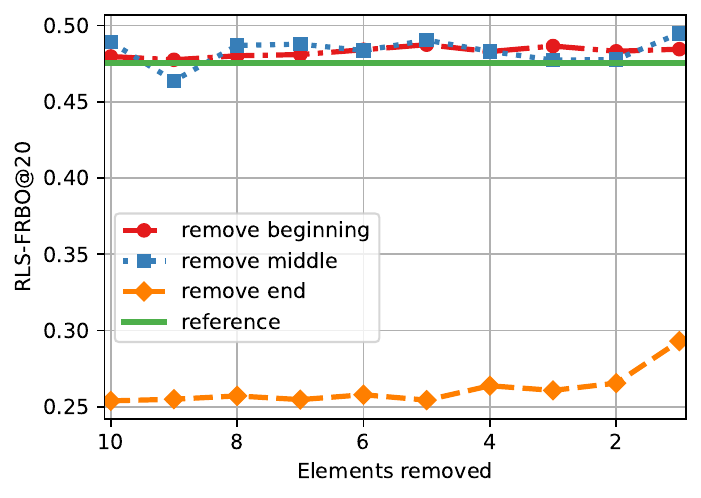}
          \caption{FRBO@20 SASRec FS-TKY}
          \label{fig:datasets_comparison_rbo_tky}
     \end{subfigure}
     \hfill
     \begin{subfigure}[b]{0.47\textwidth}
         \centering
         \includegraphics[width=0.9\textwidth]{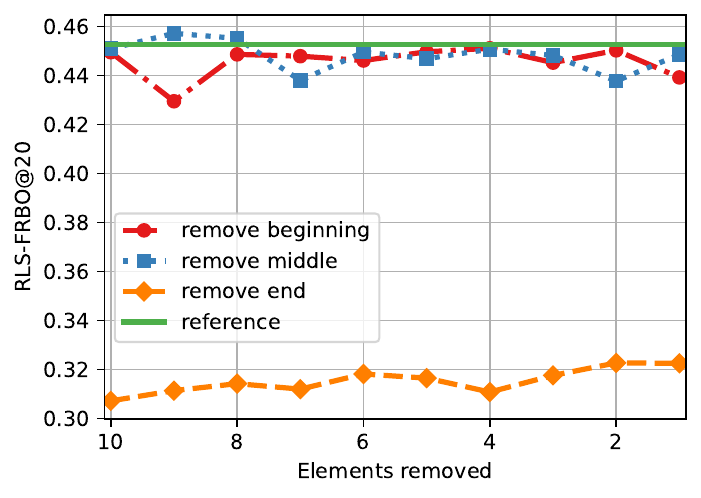}
          \caption{FRBO@20 SASRec FS-NYC}\label{fig:datasets_comparison_rbo_nyc}
     \end{subfigure}
    \caption{\looseness -1 Plots of NDCG and FRBO for SASRec on FS-TKY and FS-NYC datasets. The baseline is shown as a horizontal solid line, while dashed lines show the metrics as the number of items removed changes for the three scenarios considered.}    
    \label{fig:datasets_comparison_foursquare}
\end{figure}

Figures \ref{fig:datasets_comparison_movielens} and \ref{fig:datasets_comparison_foursquare} show the performance for the three different settings for the SASRec model applied to all datasets. From Figures \ref{fig:datasets_comparison_ndcg_ml-100k}, \ref{fig:datasets_comparison_ndcg_ml-1m}, \ref{fig:datasets_comparison_rbo_ml-100k}, \ref{fig:datasets_comparison_rbo_ml-1m} we see that the downward trend of the metric when removing items at the end of the sequence is a characteristic of the MovieLens dataset: both NDCG@20 and RLS-FRBO@20 show a decrease when increasing the number of removed items.
We hypothesize that this is happening because the average number of actions per user and the number of items (see Table \ref{tab:stats}) are not that large compared to the number of items removed.
\looseness -1 On the other hand, the Foursquare datasets (\ref{fig:datasets_comparison_ndcg_tky}, \ref{fig:datasets_comparison_ndcg_nyc}, \ref{fig:datasets_comparison_rbo_tky}, \ref{fig:datasets_comparison_rbo_nyc})  do not suffer major performance degradation, probably due to the higher average number of actions per user and the number of items (see Table \ref{tab:stats}) than MovieLens. In addition to this, and probably for the same motivation, the degradation of the RLS is lower with respect to that displayed in the MovieLens datasets.
Finally, it is interesting to note that on MovieLens 1M, there is a consistent performance degradation even when elements at the beginning and in the middle of the temporally ordered sequence are removed.
This can be observed as a small decrease in the NDCG@20 (Figure \ref{fig:datasets_comparison_ndcg_ml-1m}), but a sharp decrease in the value of the RLS-FRBO (Figures \ref{fig:number_removal_comparison_rbo_gru_ml-1m}, \ref{fig:datasets_comparison_rbo_ml-1m}). This means that even if the model performs approximately the same, the rankings produced vary greatly.
\looseness -1 The cause may be the fact that the MovieLens 1M dataset, among those considered, has the largest number of users and interactions.

\section{Conclusion} \label{sec:conclusion}
\looseness -1 In this work, we have analyzed the importance of the position of items in a temporally ordered sequence for training SRSs.
For this purpose, we introduced Finite RBO, a version of RBO for finite-length ranking lists and proved its normalization in [0,1].
Our results demonstrate the importance of the most recent elements in users' interaction sequence: when these items are removed from the training data, there is a significant drop in all evaluation metrics for all case studies investigated and this reduction is proportional to the number of elements removed. Conversely, this reduction is not as pronounced when elements at the beginning and middle of the sequence are removed.
We validated our hypothesis using four different datasets and two different models, using traditional evaluation metrics such as NDCG, Recall, but also RLS, a measure specifically designed to measure Sensitivity.
Future work in this direction could first extend our results to more models and more datasets, and then investigate a way to make the models robust to the removal of training data. We hypothesize that the solution may lie in using different training strategies \cite{petrov2022effective}, robust loss functions \cite{bucarelli2023leveraging,wani2023combining}, or different optimization objectives \cite{bacciu2023integrating}.


\bibliographystyle{splncs04}
\bibliography{citations}

\end{document}